\documentclass[a4paper,notab,USenglish, thm-restate]{lipics-v2021}

\usepackage{microtype}
\usepackage{graphicx}
\usepackage{algorithm}
\usepackage{cite}
\usepackage{booktabs}
\usepackage{todonotes}

\usepackage[T1]{fontenc}
\usepackage{amsfonts}
\usepackage{graphicx}
\usepackage{epstopdf}
\usepackage{enumitem}
\usepackage{xcolor}
\usepackage{numprint}
\usepackage{xspace}
\usepackage{hyperref}
\usepackage[noend]{algorithmic}
\newcommand{\competitorHCDSHilbert}{\texttt{Curve-H}\xspace}
\newcommand{\competitorHCDS}{\texttt{Curve-Z}\xspace}
\newcommand{\competitorEHCDSHilbert}{\texttt{Curve-H+E}\xspace}
\newcommand{\competitorEHCDS}{\texttt{Curve-Z+E}\xspace}
\newcommand{\competitorHCDSStd}{\texttt{Curve-Z (std)}\xspace}

\newcommand{\competitorZCurve}{\texttt{Curve-Z}\xspace}
\newcommand{\competitorCGALR}{\texttt{CGAL-R}\xspace}
\newcommand{\competitorCGALkd}{\texttt{CGAL-KD}\xspace}
\newcommand{\competitorBOOSTl}{\texttt{Boost-L}\xspace}
\newcommand{\competitorBOOSTq}{\texttt{Boost-Q}\xspace}
\newcommand{\competitorBOOSTR}{\texttt{Boost-R*}\xspace}
\newcommand{\competitorPAM}{\texttt{PAM-R}\xspace}
\newcommand{\competitorPKD}{\texttt{PkD}\xspace}
\newcommand{\cpp}{\textsf{C}\texttt{++}\xspace}

\usepackage{amsopn}
\newtheorem{problems}{Problem Statement}

\title{\Large Simpler is Faster:  \\
Practical Distance Reporting by Sorting Along a Space-Filling Curve}
\titlerunning{Simpler is Faster}


\author{Sarita de Berg}{IT University of Copenhagen, Denmark}{debe@itu.dk}{https://orcid.org/0000-0001-5555-966X}{}
\author{Emil Toftegaard Gæde}{Technical University of Denmark, Denmark}{etoga@dtu.dk}{}{}
\author{Ivor van der Hoog}{IT University of Copenhagen, Denmark}{ivva@itu.dk}{https://orcid.org/0009-0006-2624-0231}{}
\author{Henrik Reinst\"adtler}{Heidelberg University, Germany}{henrik.reinstaedtler@informatik.uni-heidelberg.de}{https://orcid.org/0009-0003-4245-0966}{}{}
\author{Eva Rotenberg}{IT University of Copenhagen, Denmark}{erot@itu.dk}{0000-0001-5853-7909 }{}{}

\authorrunning{S. de Berg, E. T. Gæde, I. van der Hoog, H. Reinst\"adtler, and E. Rotenberg}

\Copyright{Sarita de Berg, Emil Toftegaard Gæde, Ivor van der Hoog, Henrik Reinst\"adtler, Eva Rotenberg}

\acknowledgements{This work was supported by the Carlsberg Foundation Fellowship CF21-0302 ``Graph
Algorithms with Geometric Applications'', and the VILLUM Foundation grant (VIL37507) ``Efficient
Recomputations for Changeful Problems''. Part of this work took place while Ivor van der Hoog and Eva Rotenberg were affiliated with the Technical
University of Denmark.
 \hbox{\it Henrik Renstädtler} acknowledges support by DFG grant \hbox{SCHU 2567/8-1}
}

\hideLIPIcs
\nolinenumbers
\begin{document}

\maketitle

\begin{abstract} 
Range reporting is a classical problem in computational geometry. A (rectangular) reporting data structure stores a point set~$P$ of~$n$ points, such that, given a (rectangular) query region~$\Delta$, it returns all points in~$P \cap \Delta$. A variety of data structures support such queries:
\begin{itemize}[noitemsep]
\item $k$-d trees support queries in $O(\sqrt{n} + |P \cap \Delta|)$ expected time,
\item range trees support queries in $O(\log^2 n + |P \cap \Delta|)$ worst-case time,
\item $R$-trees support queries in $O(\log n + |P \cap \Delta|)$ expected time.
\end{itemize}

A common variant of range queries are  \emph{distance reporting queries}, where the input is a query point~$q$ and a radius~$\delta$, and the goal is to report all points in~$P$ within distance~$\delta$ of~$q$. Such queries frequently arise as subroutines in geometric data structures. Practical implementations typically answer distance queries through rectangular range queries using the data structures listed above.

This paper revisits a simple and practical heuristic for distance reporting, originally proposed by Asano, Ranjan, Roos, and Welzl (TCS '97): sort the input point set~$P$ along a space-filling curve. Queries then reduce to scanning at most four contiguous ranges along the sorted curve.

The fact that sorting along a space-filling curve is beneficial for range reporting is well-known. Many implementations use this technique to speed up their practical query and construction times. The point that this paper makes is subtle, but interesting: 
we argue that often-times, it is the space-filling curve rather than the overall data structure that provides the performance benefits. 
Thus, we offer a simple but effective alternative: only sort $P$ along a space-filling curve instead.

We compare this approach to eight range searching implementations, across an elaborate test suite of real-world and synthetic data. Our experiments confirm this simple 200-line code approach out-performs all high-end implementations in terms of space usage and construction time. 
It presents almost always the best query times.
In a dynamic setting, our approach dominates in performance. 
\end{abstract}

\newpage

\section{\texorpdfstring{Introduction \\}{Introduction.}}
\emph{Distance reporting queries} consist of a point $q$ and a distance $\delta$. The goal is to output all points of the input $P$ within distance $\delta$ of $q$. 
This is a classical problem that frequently arises in application domains. In recent years,  several implementation papers use a distance reporting data structure as a subroutine~\cite{Ali2018Maximum,buchin2025roadster,Gudmundsson2021Practical, Gudmundsson2023Practical,Kim2024SGIR, Hoog2025Efficient, Yershova2007Improving}.
We show a simple technique for efficient and practical distance reporting. We compare our technique to state-of-the-art data structures to show that our approach is not only simpler, but even competitive in a static setting and more efficient dynamically. 
In full generality, our theoretical results apply to $\mathbb{R}^d$, but we focus our implementation and experiments on the plane. 

Distance reporting queries can be reduced to range reporting queries, and all related work primarily considers range queries. Under the $L_1$ or $L_\infty$ metric, any rectangle reporting data structure can answer distance queries. Under the $L_2$ metric, one can either reduce the problem to rectangular distance queries in $\mathbb{R}^3$, or, query with a bounding square and filter the results. Whilst the latter technique has no asymptotic performance guarantees, it is common in practice because range queries become inefficient as the dimension increases.

\subparagraph{Rectangular range reporting in theory. }
A wide body of literature exists on  (planar) range reporting.  Let $k$ denote the output size of a (rectangular) range query.
The classical \emph{range tree} by Bentley~\cite{bentley1979multidimensional} is a two-level data structure: the input points are first stored in a balanced binary search tree sorted by $x$-coordinate. For each internal node, a secondary tree is built over the $y$-coordinates of the points in its subtree. This structure supports orthogonal range queries in worst-case $O(\log^2 n + k)$ time, and updates in $O(\log^2 n)$ time. The range tree requires $O(n \log n)$ space.
The \emph{k-$d$ tree}, proposed by Bentley and Saxe~\cite{bentley1975multidimensional}, is a binary search tree that alternately splits space along the coordinate axes. It is simple and practical, with expected query time $O(\sqrt{n} + k)$ in two dimensions~\cite{friedman1977algorithm}. However, its worst-case performance is linear in $n$.
The \emph{quadtree}, introduced by Finkel and Bentley~\cite{finkel1974quad}, recursively partitions the plane into four quadrants. 
We discuss more related work on quadtrees (and Morton encodings) in Section~\ref{sec:discussion_quadtree}.
The \emph{$R$-tree}~\cite{guttman1984r} offers the most practical performance. It is  a balanced, hierarchical data structure widely used in spatial databases that recursively groups objects into minimum bounding rectangles, forming a tree from bottom to top. $R$-trees offer an expected query and update time $O(\log n + k)$, but their  worst-case query time is $O(n)$.

\subparagraph{Optimising $R$-trees.}
Due to their practical performance, numerous strategies have been proposed to improve the efficiency of $R$-trees. Intuitively, each node in an $R$-tree stores a bounding rectangle that encloses several child rectangles. When a node contains too many children, it is split into two new nodes using a specified \emph{split strategy}.
Different $R$-tree variants implement different heuristics for choosing the split:
\emph{Linear $R$-trees} select two entries that are farthest apart and they guarantee that these two elements are in different new nodes (partitioning the remaining children across these two new nodes). 
This strategy typically has the fastest construction and update times, but yields poorer query performance.
\emph{Quadratic $R$-trees} choose the pair of children whose combination would cause the most overlap, attempting to separate them. This improves query efficiency over the linear variant at the cost of increased insertion time.
The \emph{$R^*$-tree}~\cite{beckmann1990r} uses a more complex strategy that includes reinserting entries when a good split cannot be found. Although more costly to build and update, this variant improves query performance.

\subparagraph{Space-filling curves.}
A \emph{space-filling curve} (SFC) is a bijective mapping from a one-dimensional interval to a higher-dimensional space. An SFC is \emph{locality-preserving} if consecutive points are nearby in the metric space. The classic example is the Hilbert curve~\cite{hilbert1891}.
Locality-preserving SFCs can improve both the static and dynamic performance of $R$-trees. During construction, the input can be sorted along the SFC, and then inserted into the $R$-tree in that order. Because consecutive nodes along the curve are nearby, they tend to be placed into the same leaf nodes. This leads to smaller, tighter minimum bounding rectangles, which improves both query and update efficiency. Additionally, the spatial coherence induced by the SFC improves cache performance for the construction and queries.

\subsection{Implementations of range searching data structures }

We compare our implementation to eight existing range-searching data structures (Table \ref{tab:implementations}).
From \texttt{CGAL}~\cite{cgal}, we include their range tree, $k$-d tree, and $R$-tree implementations.
From the \texttt{Boost} geometry library~\cite{boostgeometry}, we include its three highly optimized $R$-tree variants (linear, quadratic, and $R^*$ packing), which all three internally use Hilbert-curve–based packing to speed up their computations.  Due to its performance, \texttt{Boost} serves as our baseline. 
Kriegel and Schiwietz~\cite{Kriegel1989Performance} benchmarked in 1989 the state-of-the-art range reporting approaches. This result is too old to be competitive, and we do not compare against it, but it provided the framework for synthetic data testing that was adopted by future works~\cite{Arge2004Priority,SunBlelloch2019ALENEX,Wang2022ParGeo,Qi2020Packing}.

Beyond standard libraries, we consider parallel or research implementations from the literature. Sun and Blelloch~\cite{SunBlelloch2019ALENEX} propose the \texttt{PAM} library and its augmented \texttt{range tree}, designed for fast parallel and sequential range queries.
The authors of~\cite{Wang2022ParGeo} introduce \texttt{ParGeo}, whose $k$-d tree focuses on parallel scalability.
Men, Shen, Gu, and Sun~\cite{Men2025SIGMOD} present the \texttt{PkD-tree}, a recent parallel $k$-d tree variant.
We include these three parallel data structures, evaluating their sequential performance for a fair comparison.
Note that two of these papers~\cite{SunBlelloch2019ALENEX, Men2025SIGMOD} claim to have competitive to strong sequential performance. 

We mention two additional works that focus on I/O operations rather than runtime. Arge, de Berg, Haverkort, and Yi~\cite{Arge2004Priority} provide an $R$-tree implementation that is designed to minimise the number of I/O operations performed. Qi, Tao and Chang~\cite{Qi2020Packing} follow up this work, providing several $R$-tree implementations designed to minimise the number of I/O's used. The experimental setup in these papers is different, as they measure the number of block I/O operations their data structures use (using a block size of 4MB). They do not optimise for runtime performance, and are not competitive on the runtime metric. 

\begin{table}[h]
\centering
\begin{tabular}{@{}lll@{}}
\toprule
\textbf{\texttt{Name}} & \textbf{Brief description} & \textbf{Citation} \\ 
\midrule

\texttt{Boost-L} & Boost’s $R$-tree using linear packing. & \cite{boostgeometry} \\

\texttt{Boost-Q} & Boost’s $R$-tree using quadratic packing. & \cite{boostgeometry} \\

\texttt{Boost-R*} & Boost’s $R$-tree using $R^*$ packing. & \cite{boostgeometry} \\

\texttt{CGAL-Kd} & CGAL’s $k$-d tree for range queries. & \cite{cgal} \\

\texttt{CGAL-RTree} & CGAL’s implementation of an $R$-tree. & \cite{cgal} \\

\texttt{PAM} & Augmented range tree. & \cite{SunBlelloch2019ALENEX} \\

\texttt{ParGeo} & Parallel $k$-d tree. & \cite{Wang2022ParGeo} \\

\texttt{PkD-tree} & Parallel $k$-d tree. & \cite{Men2025SIGMOD} \\

\bottomrule
\end{tabular}
\caption{Range-searching implementations included in our comparison.}
\label{tab:implementations}
\end{table}

\subsection{Technical contribution}

Asano, Ranjan, Roos, and Welzl~\cite{ASANO19973} study point sets in $\mathbb{Z}^2$ that lie within some axis-aligned bounding box $D \subset \mathbb{Z}^2$. They define a \emph{recursive space-filling curve} (RSFC) as a bijection between $[0, \#D] \cap \mathbb{Z}$ and $D$, where $\#D$ denotes the number of grid points in~$D$.  
A recursive space-filling curve must satisfy two conditions:  
(i) $D$ can be partitioned into four disjoint equal-sized squares such that the inverse mapping~$\pi$ of the recursive space-filling curve maps each square to a contiguous interval; and  
(ii) the restriction of the RSFC to each square is itself an RSFC.  
They prove (Figure~\ref{fig:technique_example}) that for any recursive space-filling curve in $\mathbb{R}^2$ and any query square~$Q$, there exist at most $\mu \le 4$ squares $H_1,\ldots,H_\mu$ such that  
(a) the union of the $H_i$ covers $Q$;  
(b) the total area of the $H_i$ is at most $2\mu$ times the area of~$Q$; and  
(c) the inverse image of each $H_i$ is a contiguous interval.  
Although they briefly remark on the practical potential of this property, their contribution is theoretical: they show the existence of a recursive space-filling curve with $\mu = 3$, which they prove is optimal. Haverkort~\cite{HaverkortSFC} generalized many recursive space-filling curves to $d$-dimensional space.

We provide a moderate simplification and generalization of the framework of~\cite{ASANO19973} similar to~\cite{HaverkortSFC}, but streamlined for our setting. 
Consider any RSFC in $\mathbb{Z}^d$ and let $T$ be the complete quadtree over a hypercube domain $D \subset \mathbb{Z}^d$, where the leaves of $T$ are all unit hypercubes centered at integer coordinates.  
We show that a RSFC can equivalently be viewed as a function assigning to each quadtree cell an integer label such that the label of each parent is a prefix of the labels of its children.  
Let $\pi$ denote the induced mapping from $D$ to $[0,\#D] \cap \mathbb{Z}$.  
For any quadtree cell, all points it contains receive values under $\pi$ that share the same prefix.  
We prove that for any hypercube~$Q \subset D$, a RSFC generates at most $\mu = 2^d$ hypercubes $H_1,\ldots,H_{2^d}$ satisfying properties analogous to those in~\cite{ASANO19973}: $Q$ is covered by the $H_i$, the total volume of the $H_i$ is at most $2^{d+1}$ times that of~$Q$, and each $H_i$ maps to a contiguous interval.

This yields a simple data structure: store $P$ in an array sorted by $\pi(p)$.  
Given a distance query $(q,\delta)$, our (imaginary) quadtree yields hypercubes $H_1,\ldots,H_{2^d}$ that cover the query hypercube~$Q$.  
Since $\pi(H_i)$ is a contiguous interval, the points $H_i \cap P$ correspond exactly to a contiguous subarray.
Thus, each query reduces to scanning at most $2^d$ disjoint intervals.  
Our data structure therefore consists solely of an array.  
Since $\pi$ is static, the structure can be made dynamic by storing $P$ in a B-tree instead (see Section~\ref{sec:dynamic}).

\subparagraph{On novelty.}
We emphesise that idea of using space-filling curves for range reporting is not new.  
Tropf and Herzog~\cite{tropf1981multidimensional} already described range-reporting techniques on data sorted by Morton codes.  
Across multiple communities—spatial databases, computer graphics, and scientific computing—numerous works exploit space-filling curves to improve spatial indexing and traversal efficiency~\cite{Karras2012Octrees, lewiner2010fast, Morrical2017Parallel, sugiura2023zordered}.  
Furthermore, Qi, Tao and Chang~\cite{Qi2020Packing} explicitly evaluate different SFC-based packing strategies for $R$-trees. Our claim to novelty is therefore sublte:

Prior works consistently uses SFC orderings as a component inside more elaborate data structures such as $R$-trees, octrees, or pointerless spatial trees.  
We argue that, for distance queries, these additional structures are unnecessary and using only the ordering is sufficient. 
To the best of our knowledge, after an extensive literature search, we believe that we are the first practical implementation that treats a SFC-orderings a stand-alone structure.

\begin{figure}[h]
    \centering
    \includegraphics[width = \linewidth]{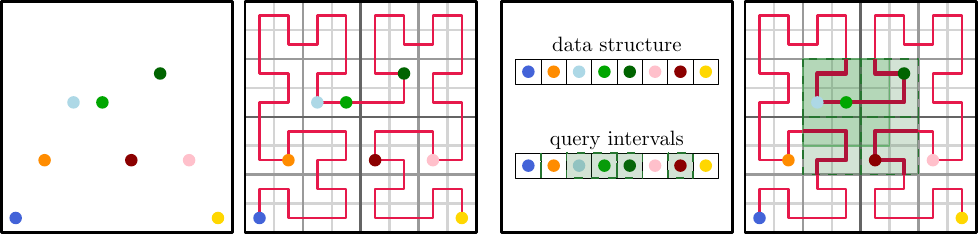}
    \caption{A Hilbert curve visits quadtree cells in a fixed order, ordering the input.}
    \label{fig:technique_example}
\end{figure}

\newpage

\subsection{Experiments and results}
We implement both planar static and dynamic versions of our technique, using the Hilbert and Z-curves. 
Our Z-order implementation generalises to $\mathbb{R}^d$ via $d$-dimensional Morton encodings. 
We compare our method against the state-of-the-art range reporting structures.

\subparagraph{Results and contribution.}
Our approach is deliberately simple and can be implemented in fewer than 200 lines of code, offering a practical advantage over more elaborate state-of-the-art structures. Using a range of synthetic and real-world data sets, we evaluate both the static and dynamic performance of modern distance-query structures (see Table~\ref{tab:implementations}). Preliminary experiments indicate that \texttt{CGAL-Kd}, \texttt{CGAL-RTree}, and \texttt{ParGeo} perform poorly compared with the remaining implementations, and we therefore exclude them from further evaluation.

In the static setting, where no points are inserted into or deleted from the data structure, we compare our approach to the state-of-the-art in terms of construction time, query time, and memory usage. We conclude that both \texttt{PAM} and \texttt{PkD-tree} perform significantly worse than the \texttt{Boost} variants and our implementations, both in terms of construction time and query time, despite claims presented in~\cite{SunBlelloch2019ALENEX} and \cite{Men2025SIGMOD} that these implementations are competitive even in a sequential setting. In terms of construction time, our approach outperforms all other implementations, due to the simplicity of our data structure.

For query times, the comparison between our approach and the \texttt{Boost} library is more nuanced (all other implementations are again not competitive). On the synthetic data sets, our structures achieve lower (better) query times than the \texttt{Boost} variants for query-window side lengths up to 0.001 times the bounding-box size. For larger windows, the highly optimised \texttt{Boost} library becomes slightly faster. In this regime, the running time is dominated by reporting the points in range, and the single \texttt{if}-statement we use to verify whether a point lies in the query window becomes the bottleneck.
On many real-world data sets, uniformly random queries often produce empty query boxes. Because the \texttt{Boost} implementations include an early-termination mechanism for this case, they outperform our baseline implementation in this setting. However, when we augment our structure with a lightweight early-termination component, our query times become significantly better than those of all \texttt{Boost} variants on almost all real-world data sets. This augmentation increases construction overhead, but for large inputs our construction remains the fastest. We note that this augmentation is not applicable in the dynamic setting.

In a dynamic setting, where points can be inserted into and deleted from the data set, we compare our approach with the three \texttt{Boost} variants; the other implementations in Table~\ref{tab:implementations} are either not competitive or do not support dynamic updates. In contrast to the static setting, the three \texttt{Boost} variants behave quite differently: the linear variant offers faster updates, while the $R^*$ variant achieves faster queries. 
In our dynamic experiments, we first insert a set of points, then delete half of them, and finally perform a number of queries. We measure the total time required for all updates and queries, for varying point-set sizes and numbers of queries. In this setting, our approach clearly dominates: our implementation is at least three times faster than all \texttt{Boost} variants in every experiment.

We thus show a simple but practically efficient technique for distance queries. 
We match or outperform the state-of-the-art in a static setting, and outperform in a dynamic setting. 
Our experiments thereby support a perhaps surprising conclusion: whilst several sophisticated data structures use space-filling curves to accelerate their performance, the simple act of sorting the input point set $P$ along such a curve is practically sufficient.

\section{\texorpdfstring{Space-filling curves and data structure}{Space-filling curves and data structure}}
In full generality, a \emph{space-filling curve} (SFC) is a bijection from a one-dimensional interval to a higher-dimensional space. In the context of computer science, however, we deal with discrete and bounded domains due to finite precision. Following~\cite{ASANO19973}, we assume that our domain $D$ is the integer grid $[1, 2^\omega]^d \cap \mathbb{Z}^d$ ($\omega$ denotes the number of bits per coordinate).

\begin{problems}
    Our input is a point set $P \subseteq [1, 2^\omega]^d \cap \mathbb{Z}^d$. We wish to store $P$ to support \emph{distance queries}. Each query consists of a point $q$ and a radius $\delta$, and the goal is to return all points in $P \cap B_\delta(q)$, where $B_\delta(q)$ denotes the ball of radius $\delta$ centered at $q$.
\end{problems}

\subparagraph{Recursive space-filling curves.}
Since the domain $D$ is finite and discrete, an SFC can be regarded as an ordering of $D = [1, 2^\omega]^d \cap \mathbb{Z}^d$.
Asano, Ranjan, Roos, and Welzl~\cite{ASANO19973} propose a heuristic for distance queries using a \emph{recursive space-filling curve} (RSFC). We present and implement their approach. First, we show their two-dimensional definition. We then propose an equivalent, generalised formulation using quadtrees, which is easier to implement.

\begin{definition}[RSFC in~\cite{ASANO19973}]
    \label{def:old}
    Let $D$ be a square domain in $\mathbb{Z}^2$ containing $m$ grid points, where $m$ is a power of two. A space-filling curve $\sigma$ is \emph{recursive} if there exist four equal-sized squares $S_1, S_2, S_3, S_4$ such that:
    \begin{itemize}[noitemsep]
        \item The first $m/4$ points of $\sigma$ are contained in $S_1$,  
        the next $m/4$ in $S_2$, and so on.
        \item $\forall i$, the restriction of $\sigma$ to $S_i$ is itself a RSFC.
    \end{itemize}
\end{definition}

\begin{definition}[Complete quadtree]
    Let $F \subseteq \mathbb{R}^d$ be a hypercube whose diameter is a power of two. The \emph{complete quadtree} $T_F$ is defined recursively as follows:
    \begin{itemize}[noitemsep]
        \item If $F$ is not a unit cube, partition $F$ into $2^d$ equal-sized hypercubes.
        \item Recurse on each hypercube to construct the full tree.
    \end{itemize}
     A \emph{level} $T^\ell_F$ of $T_F$ consists of  hypercubes of diameter~$2^\ell$. 
\end{definition}

\begin{definition}[RSFC in $\mathbb{Z}^d$]
    \label{def:recursive_SFC}
    Let the domain $D$ be a hypercube in $[1, 2^\omega] \cap \mathbb{Z}^d$ with diameter~$2^\omega$.
    Let $F$ equal $D$, translated by $- \frac{1}{2}$ in each cardinal direction. 
    A \emph{recursive space-filling curve} $\pi$ over $D$ is a function that assigns, for all levels $\ell \in [\omega]$, to each cube $C \in T_F^\ell$ a unique bit string $\pi(C)$ of length $d \cdot \ell$, s.t.:
    \begin{itemize}[noitemsep]
        \item Each cube $C$ has a unique string $\pi(C)$, and
        \item $\forall C' \in T_F$ with $C' \subset C$, $\pi(C)$ is a prefix of $\pi(C')$.
    \end{itemize}
    Our RSFC $\pi$ assigns  $\forall p \in D$ the value $\pi(p) := \pi(C)$, where $C$ is the unique unit cube in $T_F$ containing~$p$.
\end{definition}

\begin{restatable}[Proof in the appendix]{lemma}{equal}
\label{lem:equal}
   There is a one-to-one correspondence between the orderings $\sigma$ of $D$ from Definition~\ref{def:old} and the mappings $\pi$ of $D$ from Definition~\ref{def:recursive_SFC}.
\end{restatable}

\subparagraph{Distance queries.}
Definition~\ref{def:recursive_SFC} allows for a new formulation and implementation of the distance query heuristic of~\cite{ASANO19973}.
Each point $p \in P$ is assigned a bitstring $\pi(p)$ of length $d \cdot \omega$. For any hypercube $C \in T_F$, we define $\pi(C) \circ \overline{0}$ as the string obtained by appending zeroes to $\pi(C)$ until the total length is $d \cdot \omega$; we define $\pi(C) \circ \overline{1}$ analogously.

\begin{observation}
For any $C \in T_F$ and $q \in C \cap \mathbb{Z}^d$, we have:
$
\pi(C) \circ \overline{0} \leq \pi(q) \leq \pi(C) \circ \overline{1}.
$
\end{observation}

\begin{definition}
Given a query hypercube $Q$, we define $\texttt{cells}(Q)$ as the hypercubes $C_i \in T_F$ satisfying:
 $|Q| \leq |C_i| < 2|Q|$, and  $C_i \cap Q \neq \emptyset$.
\end{definition}

Let $A_P$ be an array storing all $p \in P$, sorted by increasing $\pi(p)$. Then the distance query heuristic (see Algorithm~\ref{alg:query}) proceeds as follows:
Given a query $(q, \delta)$, compute the smallest axis-aligned hypercube $Q$ that contains $B_\delta(q)$. Compute $\texttt{cells}(Q)$ as defined above. For each $C_i \in \texttt{cells}(Q)$, perform a binary search on $A_P$ to find the subarray between $\pi(C_i) \circ \overline{0}$ and $\pi(C_i) \circ \overline{1}$. Iterate over each point $p'$ in that subarray, returning those for which $\|p' - q\| \leq \delta$.

\subparagraph{Efficiency. }
The union of all cubes in $\texttt{cells}(Q)$ covers a region whose volume is at most $2 \cdot 2^d$ times that of $Q$. At first glance, this may suggest inefficiency. In the worst case, $\Theta(n)$ points could lie within $\texttt{cells}(Q)$ even if $Q$ itself contains no points from $P$, yielding an $O(n)$ additive overhead.
However, several practical factors mitigate this cost:
First, the algorithm iterates only over the points in $P$, not the entire grid. Thus if $P$ is sparse, then $\texttt{cells}(Q) \cap P$ is likely to be small. Second, all range scans occur over contiguous memory, leading to highly efficient cache behaviour—often outperforming structures like binary search trees in practice.

\begin{algorithm}
    \caption{\texttt{Query}(point $q$, distance $\delta$, array $A_P$)}
    \label{alg:query}
    \begin{algorithmic}
    \STATE $Q \gets$ bounding box $B_\delta(q)$.  
    \STATE output $\gets \emptyset$
    \FOR{$C_i \in \texttt{cells}(Q)$}
        \STATE $x \gets$ $A_P$.\texttt{lower\_bound}($\pi(C_i) \circ \overline{0}$)
        \WHILE{ $\pi(A[x]) < \pi(C_i) \circ \overline{1}$}
        \IF{$\| A[x] - q \| < \delta$}
        \STATE output.append($x$)
        \ENDIF
        \STATE x++
        \ENDWHILE
    \ENDFOR
    \RETURN output
    \end{algorithmic}
\end{algorithm}

\subsection{Implementation details}
Let the domain $D$ and $F$ be fixed. Algorithm~\ref{alg:query} yields a simple and lightweight heuristic for answering distance queries. We now describe a few implementation details.

\subparagraph{Choice of RSFC.}
In~\cite{ASANO19973}, the authors initially consider the Hilbert curve. They then provide a theoretical construction of a recursive space-filling curve (RSFC) in two dimensions that guarantees the number of unique elements in $\texttt{cells}(Q)$ is at most three, rather than four. However, we note that although this improves the number of binary searches needed, and may reduce the total volume of the region considered, their function is complicated to implement.

The Hilbert curve is widely used in data structures because it is locality-preserving: all pairs of consecutive points along the curve have distance at most two. This property is exploited to speed up range reporting structures, including the $R$-trees in Boost~\cite{boostgeometry}. We therefore include a Hilbert-curve implementation.
Concretely, we implement the function $\pi$ from Definition~\ref{def:recursive_SFC} such that the ordering of $D$ by $\pi$ follows a Hilbert curve. The computation of the ranks is accelerated by the use of a lookup table. 

However, in our specific application, the locality-preserving property of the Hilbert curve is not exploited. Distance queries are handled purely by binary search and direct scanning. Therefore, we also offer a second implementation using the Z-curve (also known as the Morton curve).
The Z-curve has an additional practical property: for any point $q \in N$, the bit-string $\pi(q)$ can be computed efficiently by interleaving coordinate bits. That is, the bit pairs in $\pi(q)$ can be constructed independently (see Algorithm~\ref{alg:zcurve}).

\subparagraph{Storing $P$ twice.}
To accelerate queries, we store the point set $P$ in two separate data structures. The first stores the elements of $P$ in an array $A_P$ sorted by their $\pi$-values. The second stores the corresponding bit strings $\pi(p)$ for all $p \in P$ in a B-tree. A 1:1 correspondence to values in $A_P$ and values in the B-tree makes sure that we avoid computing $\pi(p)$ during queries. In particular, the binary search step ($A_P$.\texttt{lower\_bound} in Algorithm~\ref{alg:query}) becomes a direct comparison over the entries in the B-tree. Our implementation uses a custom static B-tree with AVX-512 instructions following an implementation presented in \cite{AlgorithmicaSTree} with a slight overhead to store the tree, where the lowest layer is contiguous and sorted in memory.

\subparagraph{Dynamic implementation.}
\label{sec:dynamic}
We simply store the pairs $(\pi(p), p)$ in a B-tree multimap.

\begin{algorithm}[t]
    \caption{Z-curve \texttt{$\pi$}(point $p \in [2^\omega]^2$)}
    \label{alg:zcurve}
    \begin{algorithmic}
    \STATE output $\gets$ string$[2 \omega]$  
    \FOR{$i \in [ \omega]$}
        \STATE output$[2i] \gets $ the $i$'th bit of $p.x$
        \STATE output$[2i + 1] \gets$ the $i$'th bit of $p.y$
    \ENDFOR
    \RETURN output
    \end{algorithmic}
\end{algorithm}

\subsection{Quadtrees and early termination}
\label{sec:discussion_quadtree}

Algorithm~\ref{alg:zcurve} implements what is commonly known as the Z-curve, or Morton encoding. These are widely used in quadtrees, as they provide an efficient means to determine, for any point~$q$ and quadtree level~$\ell$, a bit string that uniquely identifies the hypercube at level~$\ell$ containing~$q$. This property enables the use of quadtrees without pointers~\cite{Chang1995Parallel, Sundar2008Bottom, tang2001linear,zhang2013inverted, zhang2019efficient}.

A quadtree supports distance reporting queries as follows: Given a query $(q, \delta)$, it computes the smallest axis-aligned hypercube~$Q$ that contains the ball~$B_\delta(q)$. It then identifies all nodes at level~$\lceil \log \|B_\delta(q)
\| \rceil$ that intersect~$Q$. For each such node~$C$, it traverses its descendant leaves and reports the points in those leaves that lie within~$B_\delta(q)$. By Definition~\ref{def:recursive_SFC}, Lemma~\ref{lem:equal} and Algorithm~\ref{alg:query}, we thereby show that process is equivalent to the heuristic proposed by~\cite{ASANO19973}. 

There is a crucial distinction: our quadtree is purely conceptual. We neither construct nor traverse a tree. Instead, we operate solely on a sorted array. While quadtree distance queries follow the same high-level logic, they do so in an involved and less efficient manner:
\begin{itemize}[noitemsep]
\item Regular quadtrees incur significant overhead because leaves are not stored contiguously.
\item Pointer-free quadtrees use hash maps or dictionaries to map Morton codes to quadtree nodes~\cite{Sundar2008Bottom, zhang2013inverted, zhang2019efficient}, which introduces a significant lookup overhead.
\item Karras~\cite{Karras2012Octrees} lays out quadtree cells in an array such that Morton codes directly address the corresponding quadtree cells. However, this approach still stores the entire tree. In contrast, we store only the leaves: i.e., the point set~$P$ sorted along the Z-curve.
\end{itemize}
\noindent
In essence, our queries execute the same logic as quadtree queries. As such, we argue that a performance comparison against quadtrees primarily reflect implementation overheads, not fundamental algorithmic differences. 
In addition, quadtrees are typically implemented to support nearest-neighbour queries. We are not aware  of any implementation that claims to outperform $R$-trees. Typically, $R$-trees have better performance than quadtrees~\cite{Kothuri2002Quadtree}, which is also illustrated by \texttt{Boost} using $R$-trees.
We thus do not include quadtrees in our experiments.

\subparagraph{Early termination.}
To improve performance on empty query ranges, one can construct a separate early-termination data structure based on a quadtree. Such a data structure performs a preliminary check on the query range to determine whether it is empty or not. We implement an early-termination data structure by taking a bounding box $G$ of the data whose side length is a power of two, and then constructing a quadtree by splitting $G$ into four equal-sized squares up to some constant depth $d$.
This generates a tree of depth $d$. 
For all nodes in the tree, we store a boolean indicating whether the cell contains a point or not. 
Given a query, we compute the common prefix and find the cell's boolean in constant time. If its boolean is \texttt{false}, then we simply return that the query range is empty; otherwise, we run our regular range searching algorithm.
In our implementation, we set the constant $d = 12$.

\section{\texorpdfstring{Experiments \\}{Experiments}}\label{sec:experiments}
We conduct a range of experiments, where we compare our algorithm against a wide variety of algorithms ranging from standard solutions from \texttt{Boost} and \texttt{CGAL}, to  several recent parallel libraries in sequential mode. 

At a high level, we perform four types of experiments.
First, we measure construction and query performance on synthetic data. Synthetic data enables controlled experiments at larger input sizes and facilitates asymptotic analysis. Second, we measure performance on real-world data. Following established precedent in range reporting~\cite{Qi2020Packing}, we evaluate all solutions on the real-world \textsc{Tiger} and \textsc{TigerEast} data sets.
Additionally, we consider an application domain: subtrajectory clustering~\cite{Hoog2025Efficient}, where distance queries dominate the running time. This allows us to assess performance in a realistic setting,  and it allows us to use additional real-world data: the trajectory data used in a clustering application.
Finally, we evaluate dynamic behaviour on synthetic data by, for varying input sizes and number of queries, first inserting points, then deleting half of them, and finally performing the queries.

\subparagraph{Methodology.} We carry out our implementation in \cpp and compile our programs with \texttt{gcc-11.4} with full optimization flags enabled. The machines used are two identical equipped with  a Xeon w5-5-3435X running at 3.10 Ghz having a L3 cache of 45 MB and 128 GB RAM.
Results are not compared across machines and 
each experiment is repeated 3 times. Up to 4 experiments are scheduled in parallel if the main memory consumption permits it. The order of the experiments is randomized for fairness.
\subparagraph{Algorithms.}
\begin{table}[t]
\centering
\begin{tabular}{@{}lll@{}}
\toprule
\textbf{\texttt{Name}} & \textbf{Brief description} & \textbf{Citation} \\ 
\midrule

\texttt{Boost-L} & Boost’s $R$-tree using linear packing. & \cite{boostgeometry} \\

\texttt{Boost-Q} & Boost’s $R$-tree using quadratic packing. & \cite{boostgeometry} \\

\texttt{Boost-R*} & Boost’s $R$-tree using $R^*$ packing. & \cite{boostgeometry} \\

\texttt{CGAL-Kd} & CGAL’s $k$-d tree for range queries. & \cite{cgal} \\

\texttt{CGAL-RTree} & CGAL’s implementation of an $R$-tree. & \cite{cgal} \\

\texttt{Curve-Z+E} & Our algorith using a $Z$-curve and early termination. & This paper \\

\texttt{Curve-H+E} & Our algorith using a Hilbert-curve and early termination. & This paper \\

\texttt{Curve-Z} & Our algorith using a $Z$-curve without early termination. & This paper \\

\texttt{Curve-H} & Our algorith using a Hilbert-curve without early termination. & This paper \\

\texttt{PAM} & Augmented range tree. & \cite{SunBlelloch2019ALENEX} \\

\texttt{ParGeo} & Parallel $k$-d tree. & \cite{Wang2022ParGeo} \\

\texttt{PkD-tree} & Parallel $k$-d tree. & \cite{Men2025SIGMOD} \\

\bottomrule
\end{tabular}
\caption{Range-searching implementations included in our comparison.}
\label{tab:implementations2}
\end{table}
We implement our algorithm using two space-filling curves, \competitorZCurve uses a Z-curve, and \competitorHCDSHilbert uses a Hilbert curve. Additionally, we implement an early-termination data structure, and include it for the respective SFCs in \competitorEHCDS and \competitorEHCDSHilbert.
For the dynamic version, we use the B-tree multimap implementation of the Abseil team~\cite{AbseilCPP} and compare to the one provided by the standard library.
All implementations available in our repository are listed in Table~\ref{tab:implementations2}.
\competitorPAM, \texttt{ParGeo} and \competitorPKD present parallel libraries.
However, two of these also claim strong sequential performance. For fairness of comparison, we evaluate all algorithms based on their sequential performance.
We also include \texttt{CGAL}s~\cite{cgal} implementation of $k$-d trees and $R$-trees in our benchmark suite. We note that the implementation table  excludes the data structures from~\cite{Arge2004Priority, Qi2020Packing}. This is because the data structures from~\cite{Arge2004Priority, Qi2020Packing} optimise for I/O's rather than runtime, and they are not competitive on runtime.

Preliminary experiments showed that \texttt{ParGeo}~\cite{Wang2022ParGeo} and \texttt{CGAL}~\cite{cgal} are too slow during querying, requiring more than 1 second to answer 1000 queries. The construction times of \texttt{ParGeo} were competitive. The results of this preliminary experiment can be found in Appendix~\ref{app:prelim}.
Therefore, we exclude these two implementations from extensive testing.

\subparagraph{Data.} Our data set comprises real-world data and synthetically generated inputs from several distributions. 
We use~32-bit unsigned integers as coordinates in all experiments in this paper.
Our code also works on 16-bit integers and is easily extensible to other widths.

The real-world instances are listed in Table~\ref{tab:realworld:instances}. The points in these sets are uniformly scaled  to span $\Delta=2^{31}$ in the larger dimension. 
We follow precedent~\cite{Qi2020Packing} and use the \textsc{tiger2006se} data set~\cite{tiger2006se}, which consists of geographical features from the US Census Bureau's system. 
We also use their smaller \textsc{tiger2006east} data set.
We additionally evaluate performance on several real-world data sets commonly used in trajectory analysis. The map-construction data sets---\textsc{Athens\_Large}, \textsc{Athens\_Small}, \textsc{Berlin\_Large}, and \textsc{Chicago}---consist of GPS traces from urban car traffic and serve as standard benchmarks for road-map reconstruction~\cite{ahmed2015comparison, buchinGroup2017, wang2015efficient, huang2018automatic}.
The fourth data set is \textsc{Drifter} from the NOAA Global Drifter Program~\cite{gdacDataset}. Following~\cite{Hoog2025Efficient}, we use all drifters recorded from 2022--2024, yielding roughly $5{,}500$ trajectories with an average length of $500$ points (in total $n \ge 2 \cdot 10^6$).  
Finally, \textsc{UnID}~\cite{9827305} consists of trajectory data collected on German highways.

To generate the synthetic data sets, we employ three distributions. For a window size of $\Delta=2^{32}$, the distributions are defined as follows:
\begin{itemize}
    \item {\bf Uniform:} A uniform distribution across the whole plane.
    \item {\bf Normal:} A normal distribution centered on the plane $\frac{\Delta}{2}$ with a standard deviation of $\frac{\Delta}{8}$ similar to parameters in~\cite{Kriegel1989Performance}.
    \item {\bf Skewed:} Uniform sampling along the $x$-axis and geometric distribution with $p=\frac{5}{\Delta}$ along the $y$-axis. These parameters are from \cite{Qi2020Packing}.
\end{itemize}

Dynamic data is generated as follows.
First, we sample $n$ points from a distribution. Second, we delete half of these points by random selection and in random order. Afterwards we query $k\cdot n$ times, for $k \in\{1,1.5,10\}$ following a (different) distribution.
We use 5 instances with different seeds for static and dynamic data sets per~configuration.

\subparagraph{Queries.}
For the queries we again use the uniform, normal, and skewed distributions defined before. We use a square query window with a relative side length $\delta$ between $10^{-7}$ and~0.02.
That is, each query has a size length $\delta$ times the bounding box bounding the domain. 

\begin{table}
    \centering
    \begin{tabular}{lr}
    \toprule
Name&\# Points\\
\midrule
      \textsc{athens\_large}&   72439 \\
      \textsc{athens\_small}&2840 \\
 \textsc{berlin\_large}&   192223 \\
 \textsc{chicago} &     118360 \\
 \textsc{drifter}&    1792084 \\
 \textsc{unid\_00}&     214076 \\
  \textsc{tiger2006east} & 17468292 \\
 \textsc{tiger2006se} & 35904949 \\\bottomrule
    \end{tabular}
    \caption{Statistics for the realworld data sets.}
    \label{tab:realworld:instances}
\end{table}
\subsection{Static synthetic experiments}
We investigate the scaling properties of our algorithms and competitors.
Firstly, we vary the number of points stored in the data structures.
Secondly, we study the influence of the relative query window size~$\delta$.
The choice of distributions does not greatly influence the results on synthetic data.
Accounting for that and the limited space, we limit our discussion here to point generation by normal distribution and uniform query sampling. We refer to Appendix~\ref{app:other:dist} for the plots and tables containing the results of all distributions under uniform~sampling.
\subsubsection{Varying input size}
In this experiment, we fix the query window size to $\delta=0.001$ and the number of queries to~$10^6$. We report on construction time, query time and memory consumption.

\subparagraph{Construction time.} In Figure~\ref{fig:exp:norm:unif:scaling:constructionperpoint} the construction time per point is shown. 
Naturally, our \competitorHCDS approach is the fastest, since we only need to compute the bit strings $\pi(A[i])$ and sort them once. For $10^9$ points this is 1.76 times faster than \competitorBOOSTl.
Our approaches that additionally construct an early-termination data structure (\competitorEHCDS and \competitorEHCDSHilbert) require more time per item, when there are fewer points to construct from. For every point we have to write into $2*d = 24$ cells (one per level) of the quadtree.
The spike around $10^6$ can be explained by the following:
When we have more than $2^{24}/24$ points, we write the points into the finest level and compute the coarser levels by combining the neighbouring cells. 
For $10^6$ this leads to slightly higher costs per point.
For larger input sizes, the cost of building this additional data structure is amortized and for inputs of size $10^9$ this extra data structure costs 5.3\% more time.
\competitorPAM segfaults for $10^8$ and $10^9$ and takes 12 times more time for $10^7$.
We suspect that there is a implementation error in it.
\competitorPKD takes twice the time of our approaches.
  Table~\ref{tab:app:const:scaling} in the Appendix reports the absolute timings for this experiment.
  
\subparagraph{Query Time.}
The query time results are shown in Figure~\ref{fig:exp:norm:unif:scaling:query} and detailed results can be found in Table~\ref{tab:app:query:scaling}.
For low density point sets ($\leq10^6$ points) our \competitorEHCDS approach is the fastest, being nearly twice as fast as the version without the auxiliary early-termination data structure for $10^4$ points. This indicates that the early termination is indeed helpful when many of the query regions are empty.
For this lowest density our methods are up to 3 times faster than the \texttt{Boost} library variants.
The more points are in the set, the smaller the advantage of our technique becomes.
For $10^9$ points \competitorBOOSTq is the fastest method, being 7.5\% faster than~\competitorEHCDSHilbert. In this regime, the query time is dominated by the reporting of the points in the range. The additional if-statement in our approach to check if the point that is found indeed lies in the query region results in the additional overhead in this regime.

\subparagraph{Memory consumption.}
We measured the peak total amount of resident memory via \texttt{/proc/} and report the results in Figure~\ref{fig:exp:norm:unif:scaling:memper} (and Table~\ref{tab:app:mem:scaling} in the appendix). For $10^4$ and $10^5$ the short running times do not allow one to reliably capture these values. As \competitorPAM uses more than 12 times more memory than our code, it is excluded in the figure.
We suspect there is an implementation mistake in this library.
The overhead of this measurement is that we have to  store the points and queries separately once. For the larger point sets, our method requires substantially lower memory, about 23.2\% less.
For $10^9$ points there is nearly no difference between the early-termination methods and the implementations without it.
For small data sets the overhead of the early-termination data structure is relatively large, but still small in absolute terms: It requires a fixed memory of~$2^{25}$ bits, which is roughly 4.2 megabyte.

\begin{figure}
    \centering
    \includegraphics{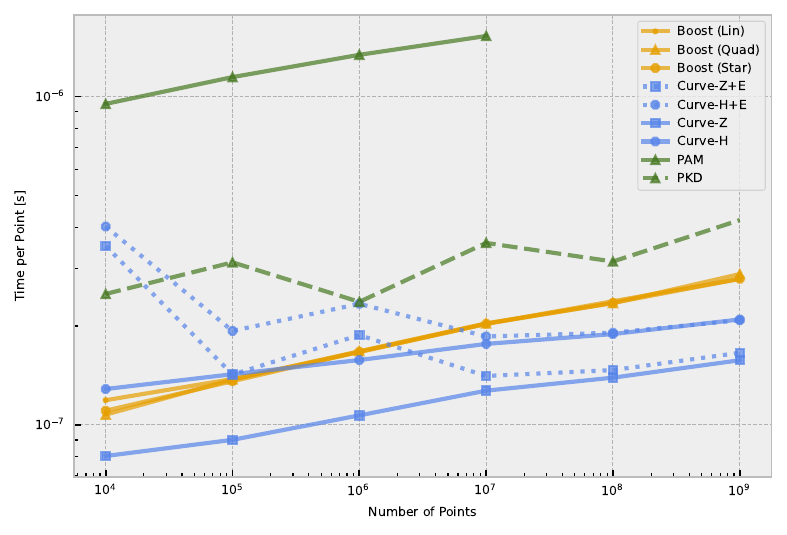}
    \caption{Time needed for constructing the data structure in seconds per point from $10^4$ to $10^9$. }
    \label{fig:exp:norm:unif:scaling:constructionperpoint}
\end{figure}
\begin{figure}
    \centering
    \includegraphics{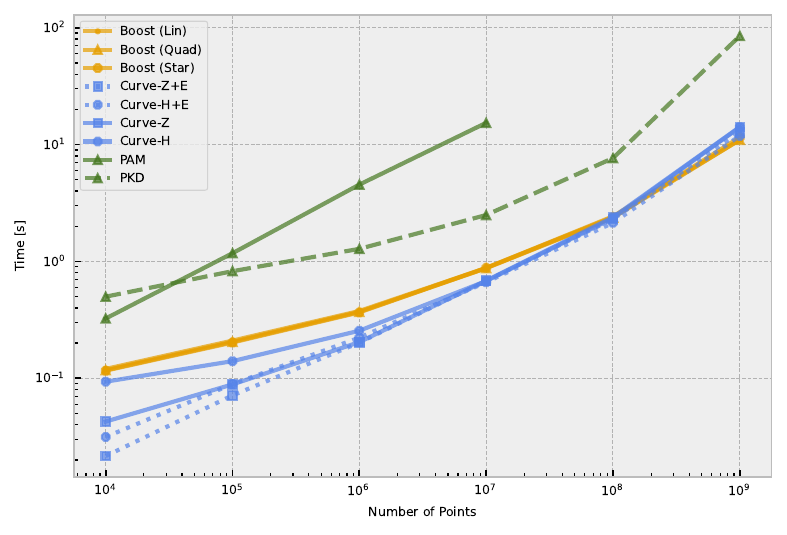}
    \caption{Overall query time to answer  $10^6$ uniform queries with size $\delta=0.001$ for $10^{4}$ to $10^9$.}
    \label{fig:exp:norm:unif:scaling:query}
\end{figure}
\begin{figure}
    \centering
    \includegraphics{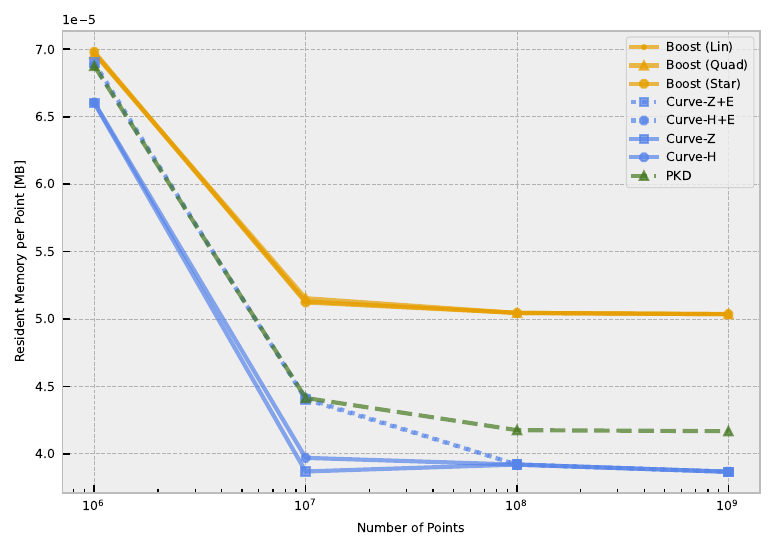}
    \caption{Maximum resident memory per point for $10^6$ to $10^9$ points reported via \texttt{/proc/}. This measure includes the queries and points once. PAM requires more memory, see Table~\ref{tab:app:mem:scaling}.}
    \label{fig:exp:norm:unif:scaling:memper}
\end{figure}
\subsubsection{Query window}
In this experiment, we explore how the size of the query window influences the query time. We fix the number of points  to $10^{7}$ and perform $10^6$ queries.

\subparagraph{Query time.} Figure~\ref{fig:exp:scalingw:query} shows the total query time plotted for the relative query window length $\delta$ from $10^{-7}$ to $0.02$.
For $10^{-7}$ every query region was empty.
Table~\ref{tab:app:query:scalingw} in the appendix contains all results. 
For large values of $\delta\geq 0.01$ the \texttt{Boost} variants are faster than our approach, about~35.8\% for~$\delta=0.02$. Again, this can again be explained by the fact that our algorithm has an additional if-statement for every point that is found to check whether it is indeed in the query region, and in this regime the query region (and the area surrounding it) contains many points. \competitorEHCDS dominates for small values of $\delta$ being $5.6$ times faster than the fastest of the \texttt{Boost} variant (\competitorBOOSTR) in the extreme case.
The early-termination reduces the running time by~37\% compared to~\competitorHCDS in this setting.

\begin{figure}
    \centering
    \includegraphics{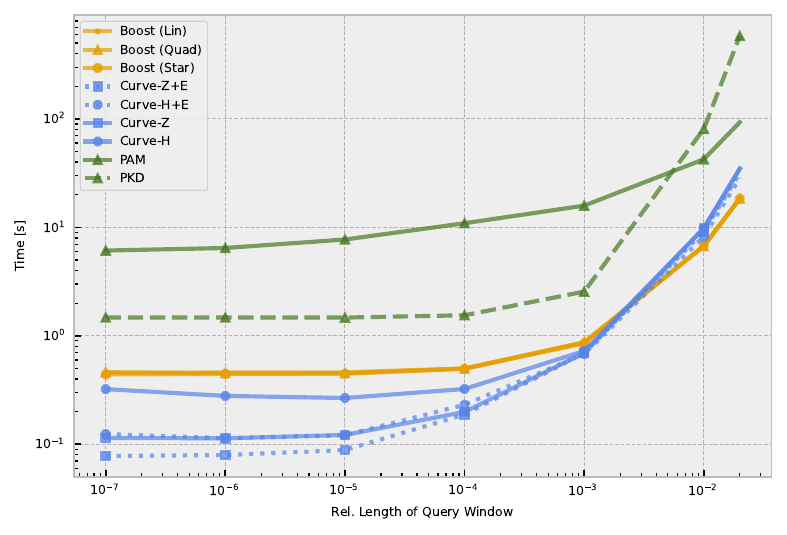}
    \caption{Plot of query time dependent on the relative length of the query window  in seconds.}
    \label{fig:exp:scalingw:query}
\end{figure}
\subsection{Real-world experiments}
In this section we discuss the performance of the algorithms on real world data. We report the construction time and query time for $10^{6}$ queries sampled according to the three distributions.

\begin{figure}
    \centering
    \includegraphics{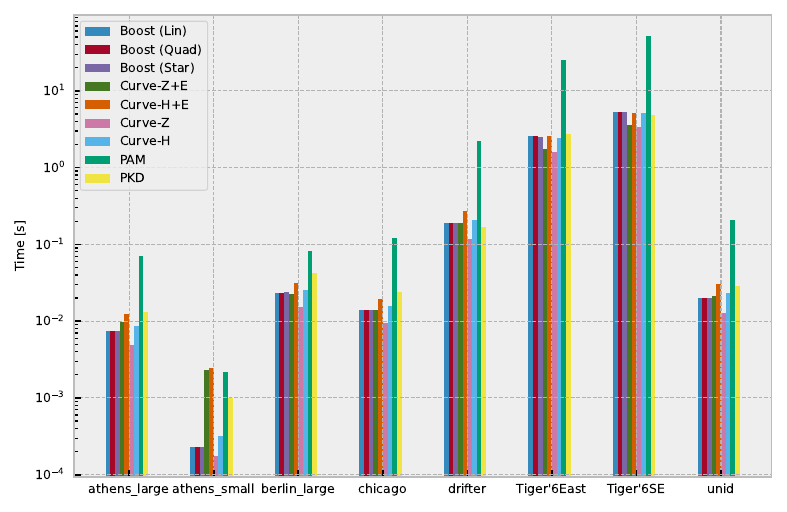}
    \caption{Average construction time for the real-world datasets. }
    \label{fig:exp:realworld:construction}
\end{figure}
\subparagraph{Construction time.}
The construction times are shown in Figure~\ref{fig:exp:realworld:construction}.
Similar to the results for synthetic data,  \competitorHCDS has the smallest construction time. 
It is 35.6\% faster than \competitorBOOSTR across all data sets on average.
On the ultra small \textsc{athens\_small} (2840 points) data set, we again see the  overhead introduced by the early-termination data structure, similar to the small synthetic data sets.
Here, \competitorEHCDS is more than 10 times slower than the version without early termination. For the large data sets \textsc{Tiger2006East} and \textsc{Tiger2006SE}, all of our algorithms outperform the three \texttt{Boost} algorithms.

\subparagraph{Query time.}
Figure~\ref{fig:exp:realworld:query} reports the query times for uniformly sampled queries. We report the results for all three distributions (uniform, normal, and skewed) in Table~\ref{tab:app:realworld:query} in the appendix.
On all data sets except for the \textsc{tiger2006se} data set either \competitorEHCDS or \competitorEHCDSHilbert is the fastest algorithm. For the ultra small sets, our implementation that include early-termination are much faster than \texttt{Boost}. For example, for \textsc{athens\_small} \competitorEHCDS is 4.6 times faster than \competitorBOOSTR and 5.25 times faster then its vanilla implementation (\competitorHCDS).
For the largest dataset (\textsc{tiger2006se} with 35 million points), the \competitorBOOSTq algorithm is the fastest for uniformly and skewed sampled queries.
\begin{figure}

    \centering
    \includegraphics{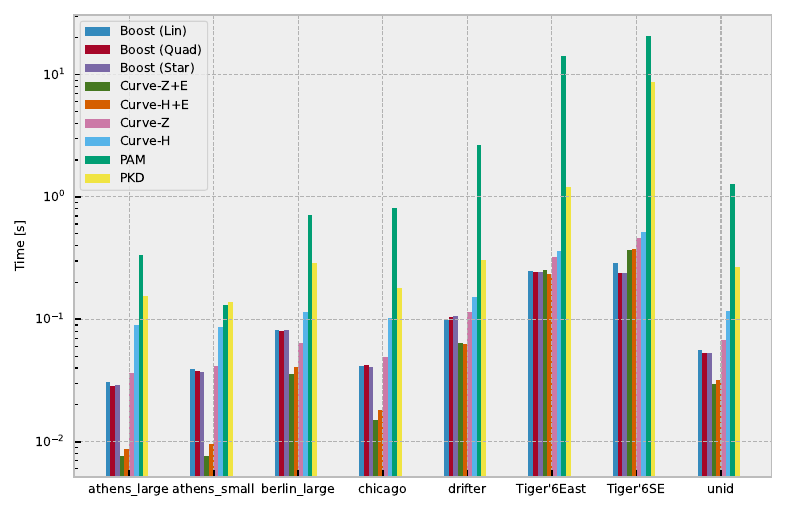}
    \caption{Average query time for $10^{6}$ queries on the real-world data sets.}
    \label{fig:exp:realworld:query}
\end{figure}

\subsection{Dynamic synthetic experiments}
The running times for dynamically inserting $n$ points, using a normal distribution, then deleting half by random sampling, and answering $k\cdot n$ uniform queries with $k\in \{1,1.5,10\}$ are listed in Table~\ref{tab:exp:dynamic:runtime}. For $k=1$ we report the results in Figure~\ref{fig:exp:dynamic:runtime}.
Our algorithm is on average three times faster than the \competitorBOOSTR in the largest query 
heavy scenario with $k=10$. For $k=1$ this ratio grows to 4.3. 
The choice of the map container is also important: \competitorHCDSStd using \texttt{std::multimap} is 5.53 times slower than \competitorHCDS for $k=10$ and $n=10^7$, which uses \texttt{absl::btree\_multimap}. On the other hand, the difference between the two space-filling curves is small, with the Z-curve always outperforming the Hilbert curve.

\begin{table}[]
    \centering
    \begin{tabular}{llrrrrrr}
\toprule
Insertions&Queries&\competitorBOOSTl&\competitorBOOSTq&\competitorBOOSTR&\competitorHCDS&\competitorHCDSHilbert&\competitorHCDSStd\\ \midrule
$10^{4}$&$10^{4}$&0.009&0.008&0.021&\bfseries 0.003&0.005&0.003\\
$10^{4}$&\numprint{15000}&0.010&0.010&0.022&\bfseries 0.003&0.005&0.004\\
$10^{4}$&$10^{5}$&0.034&0.027&0.034&\bfseries 0.012&0.017&0.013\\
$10^{5}$&$10^{5}$&0.344&0.159&0.245&\bfseries 0.038&0.054&0.070\\
$10^{5}$&\numprint{150000}&0.424&0.191&0.261&\bfseries 0.045&0.064&0.085\\
$10^{5}$&$10^{6}$&1.769&0.726&0.537&\bfseries 0.182&0.235&0.303\\
$10^{6}$&$10^{6}$&17.053&5.665&3.951&\bfseries 0.806&0.977&2.588\\
$10^{6}$&\numprint{1500000}&21.473&7.181&4.489&\bfseries 1.020&1.183&3.260\\
$10^{6}$&$10^{7}$&95.065&32.577&13.575&\bfseries 4.293&4.745&14.371\\
$10^{7}$&$10^{7}$&1375.712&253.148&99.312&\bfseries 23.067&24.434&111.824\\
$10^{7}$&\numprint{15000000}&1739.294&332.676&121.817&\bfseries 30.852&32.249&151.198\\
$10^{7}$&$10^{8}$&\em OOT&1649.575&490.485&\bfseries 162.034&163.242&896.422\\

\bottomrule
    \end{tabular}
    \caption{Average running time in seconds for the dynamic test suite (deleting 50\% of the insertions). \emph{OOT} instances did not finish with in the time limit of 1 hour.
    }
    \label{tab:exp:dynamic:runtime}
\end{table}
\begin{figure}

    \centering
    \includegraphics{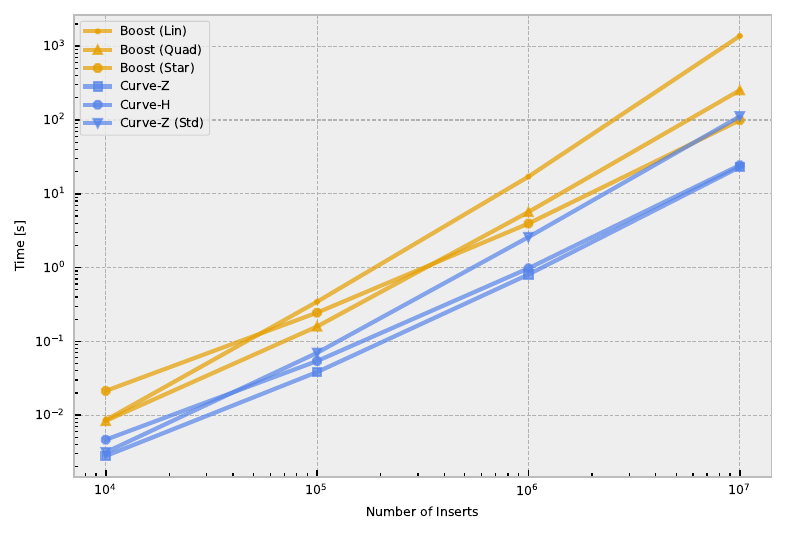}
    \caption{Running time for insertions and queries ($k=1$) of the dynamic algorithms. }
    \label{fig:exp:dynamic:runtime}
\end{figure}

\section{\texorpdfstring{Conclusion \\}{Conclusion.}}
We proposed a simple and practical data structure for distance-reporting queries, based on sorting the input point set along a space-filling curve. Although the underlying idea is not new~\cite{ASANO19973}, many recent papers favour range-reporting structures over this approach~\cite{Ali2018Maximum,buchin2025roadster,Gudmundsson2021Practical,Gudmundsson2023Practical,Kim2024SGIR,Hoog2025Efficient,Yershova2007Improving}. State-of-the-art implementations already use space-filling curves to improve performance, but still wrap the points in a traditional structure such as an $R$-tree, quadtree, or $k$-d tree. The novelty of our approach is that we propose to discard this surrounding structure entirely: storing the points $P$ directly in sorted order.

We provide a basic \cpp implementation and conduct extensive experiments comparing eight modern approaches for distance queries. Notably, and contrary to the claims in~\cite{SunBlelloch2019ALENEX}, we observe that \competitorPAM is significantly slower than its competitors in sequential settings. Among the remaining methods, only the highly optimised \texttt{Boost} library is competitive.

Statically, our simplest implementation performs competitively with \texttt{Boost}, though its relative performance deteriorates in extreme cases—when many queries are empty or when queries report a large fraction of the points. With a lightweight augmentation providing early termination on empty queries, our method becomes the preferred static solution for all query diameters up to $0.01$ times the bounding-box size.
In dynamic settings, our approach consistently outperforms all competitors in total update and query time. Although \texttt{Boost} remains the strongest alternative, our solution is between a factor three and five faster.

Overall, we demonstrate that a simple, practical technique for distance queries can match or exceed the performance of sophisticated modern structures. Our experiments show that much of the practical efficiency of existing implementations can be achieved by simply sorting along a space-filling curve.

\newpage
\bibliography{references}
\clearpage

\appendix

\section{Omitted proof}

The main body omits the proof of the following lemma:

\equal*

\begin{proof}
Consider a recursive space-filling curve $\sigma$ as in Definition~\ref{def:old}. Its recursive subdivision of $D$ into equal-sized squares induces a quadtree. We construct a mapping $\pi$ satisfying Definition~\ref{def:recursive_SFC} in a top-down fashion.
Each internal tree node has four children $S_1, S_2, S_3, S_4$ such that all points of $\sigma$ in $S_1$ precede those in $S_2$, and so on. We assign to $S_1, S_2, S_3, S_4$ the binary suffixes $00$, $01$, $10$, and $11$, respectively. The code $\pi(S_i)$ is defined as the parent’s code concatenated with the appropriate suffix. This results in a unique map $\pi$ for every $\sigma$ that satisfies Definition~\ref{def:recursive_SFC}.
Conversely, given a map $\pi$ over a $D \subset \mathbb{Z}^2$ (with diameter a power of two) that satisfies Definition~\ref{def:recursive_SFC}, we can induce a total order $\sigma$ on $D \cap \mathbb{Z}^2$ by sorting the points $p \in P$ by $\pi(p)$. This produces an ordering $\sigma$ satisfying Definition~\ref{def:old}.
\end{proof}

\section{Preliminary experiment}\label{app:prelim}
We ran a basic experiment with $10^6$ points sampled from a normal distribution and \numprint{1000} uniformly random queries to determine if algorithms are feasible. We include our implementations (\competitorHCDS,\competitorHCDSHilbert), one of the \texttt{Boost} libraries, as well the \texttt{CGAL} (\competitorCGALkd,\competitorCGALR).
The query window was set to $0.001$ and we measured construction  and query time. We limited the running time of the combined task to 60 seconds.

\subparagraph{Analysis.} The construction times are shown in Table~\ref{tab:app:prelim:const} and the query time in Table~\ref{tab:app:prelim:answering}.
Because \competitorCGALkd does not finish the combined times in 60 seconds, we  exclude it from further testing.
\competitorCGALR construction time is 37.62 times higher than \competitorZCurve, its querying is 3678 times slower than \competitorZCurve. We therefore exclude it from additional experiments.
We observe that \texttt{ParGeo} construction time is very competitive, even lower than all other implementations.
However, answering only 1000 queries takes more than a second, more than 2000 times slower than our \competitorHCDS. We therefore exclude \texttt{ParGeo} from the comparison.

\begin{table}[H]
{\small
\begin{tabular}{lrrrrrrrrr}
\toprule
Instance Name&Result Size&\competitorBOOSTl&\competitorCGALkd&\competitorCGALR&\competitorHCDS&\competitorHCDSHilbert&\texttt{ParGeo}\\
 \midrule
norm:1000000:unif:1000:4:0.001&\numprint{885}&0.176&\em OOT&4.327&0.115&0.159&\bfseries 0.095\\
norm:1000000:unif:1000:5:0.001&\numprint{922}&0.177&\em OOT&4.337&0.111&0.164&\bfseries 0.100\\
norm:1000000:unif:1000:2:0.001&\numprint{958}&0.178&\em OOT&4.328&0.111&0.162&\bfseries 0.099\\
norm:1000000:unif:1000:1:0.001&\numprint{992}&0.178&\em OOT&4.343&0.113&0.160&\bfseries 0.099\\
norm:1000000:unif:1000:3:0.001&\numprint{1061}&0.179&\em OOT&4.349&0.113&0.164&\bfseries 0.096\\

\bottomrule
\end{tabular}}
\caption{Average running time for (in seconds) to construct from $10^6$ points. Timeout of \numprint{60} seconds, marked \emph{OOT} if exceeded.}
\label{tab:app:prelim:const}
\end{table}
\begin{table}[H]
{\small
\begin{tabular}{lrrrrrrrrr}
\toprule
&Result Size&\competitorBOOSTl&\competitorCGALkd&\competitorCGALR&\competitorHCDS&\competitorHCDSHilbert&\texttt{ParGeo}\\
 \midrule
norm:1000000:unif:1000:4:0.001&\numprint{885}&0.006&\em OOT&2.402&\bfseries 6.530e-04&6.983e-04&1.458\\
norm:1000000:unif:1000:5:0.001&\numprint{922}&0.006&\em OOT&2.479&\bfseries 6.284e-04&8.522e-04&1.960\\
norm:1000000:unif:1000:2:0.001&\numprint{958}&0.004&\em OOT&2.416&\bfseries 5.957e-04&8.302e-04&1.663\\
norm:1000000:unif:1000:1:0.001&\numprint{992}&0.006&\em OOT&2.412&\bfseries 6.524e-04&7.616e-04&1.599\\
norm:1000000:unif:1000:3:0.001&\numprint{1061}&0.003&\em OOT&2.406&\bfseries 6.890e-04&7.300e-04&2.177\\

\bottomrule
\end{tabular}}
\caption{Average running time for (in seconds) to query 1000 times on $10^6$ points. Timeout of \numprint{60} seconds, marked \emph{OOT} if exceeded.}
\label{tab:app:prelim:answering}
\end{table}
\section{Additional results}\label{app:other:dist}
\subsection{Normal point generation- uniform sampling}
\begin{table}[H]
{\small
\begin{tabular}{lrrrrrrrrr}
\toprule
&\competitorBOOSTl&\competitorBOOSTq&\competitorBOOSTR&\competitorEHCDS&\competitorEHCDSHilbert&\competitorHCDS&\competitorHCDSHilbert&\competitorPAM&\competitorPKD\\
 \midrule
$10^{4}$&0.001&0.001&0.001&0.004&0.004&\bfseries 8.044973e-04&0.001&0.009&0.003\\
$10^{5}$&0.014&0.014&0.014&0.014&0.019&\bfseries 0.009&0.014&0.115&0.031\\
$10^{6}$&0.167&0.168&0.167&0.188&0.234&\bfseries 0.107&0.158&1.339&0.237\\
$10^{7}$&2.030&2.039&2.032&1.410&1.859&\bfseries 1.272&1.765&15.280&3.582\\
$10^{8}$&23.808&23.486&23.530&14.691&19.089&\bfseries 13.932&18.909&&31.445\\
$10^{9}$&278.262&287.905&278.497&165.838&208.411&\bfseries 157.403&209.523&&420.476\\

\bottomrule
\end{tabular}}
\caption{Average construction time for $10^i$ for $4\leq i\leq 9$.  Normal distribution of input points. \competitorPAM crashed for $10^8$ and $10^9$. Lowest marked \textbf{bold}.}
\label{tab:app:const:scaling}
\end{table}

\begin{table}[H]
{\small
\begin{tabular}{lrrrrrrrrr}
\toprule
&\competitorBOOSTl&\competitorBOOSTq&\competitorBOOSTR&\competitorEHCDS&\competitorEHCDSHilbert&\competitorHCDS&\competitorHCDSHilbert&\competitorPAM&\competitorPKD\\
 \midrule
$10^{4}$&0.120&0.117&0.116&\bfseries 0.022&0.031&0.043&0.093&0.324&0.499\\
$10^{5}$&0.210&0.205&0.202&\bfseries 0.071&0.089&0.089&0.140&1.174&0.824\\
$10^{6}$&0.376&0.370&0.366&\bfseries 0.202&0.223&0.204&0.255&4.541&1.280\\
$10^{7}$&0.878&0.874&0.881&0.681&\bfseries 0.663&0.684&0.679&15.334&2.493\\
$10^{8}$&2.430&2.371&2.387&2.300&\bfseries 2.139&2.389&2.340&&7.654\\
$10^{9}$&10.978&\bfseries 10.977&11.523&12.733&11.870&14.153&14.051&&84.949\\

\bottomrule
\end{tabular}
}
\caption{Average query time for answering $10^6$ queries on $10^i$ points for $4\leq i\leq 9$.  Normal distribution of input and uniform random querying. \competitorPAM crashed for $10^8$ and $10^9$. Lowest marked \textbf{bold}.}
\label{tab:app:query:scaling}
\end{table}
\begin{table}[H]
{\small
\begin{tabular}{lrrrrrrrrr}
\toprule
&\competitorBOOSTl&\competitorBOOSTq&\competitorBOOSTR&\competitorEHCDS&\competitorEHCDSHilbert&\competitorHCDS&\competitorHCDSHilbert&\competitorPAM&\competitorPKD\\
 \midrule
$10^{6}$&69.724&69.682&69.817&69.007&68.823&\bfseries 66.022&66.057&784.140&68.768\\
$10^{7}$&513.051&515.435&512.322&440.435&441.267&\bfseries 386.999&397.084&8770.376&441.507\\
$10^{8}$&5045.823&5045.052&5045.566&\bfseries 3920.887&3921.172&3921.281&3921.199&&4176.385\\
$10^{9}$&50360.043&50359.879&50366.320&\bfseries 38679.059&38679.418&38679.379&38679.116&&41684.934\\

\bottomrule
\end{tabular}}
\caption{Maximal resident memory (in MB) $10^6$ queries on $10^i$ points for $4\leq i\leq 9$. Normal distribution of input and uniform random querying. \competitorPAM crashed for $10^8$ and $10^9$. Lowest marked \textbf{bold}.}
\label{tab:app:mem:scaling}
\end{table}
\begin{table}[H]
{\small
\begin{tabular}{lrrrrrrrrr}
\toprule
$\delta$&\competitorBOOSTl&\competitorBOOSTq&\competitorBOOSTR&\competitorEHCDS&\competitorEHCDSHilbert&\competitorHCDS&\competitorHCDSHilbert&\competitorPAM&\competitorPKD\\
 \midrule
1.000000e-07&0.458&0.462&0.440&\bfseries 0.078&0.124&0.114&0.323&6.112&1.475\\
1.000000e-06&0.458&0.452&0.446&\bfseries 0.080&0.114&0.114&0.279&6.455&1.475\\
1.000000e-05&0.461&0.453&0.446&\bfseries 0.088&0.121&0.122&0.267&7.729&1.474\\
1.000000e-04&0.501&0.498&0.497&\bfseries 0.188&0.231&0.199&0.324&10.914&1.547\\
0.001&0.868&0.873&0.852&0.692&\bfseries 0.681&0.696&0.725&15.860&2.555\\
0.010&6.751&\bfseries 6.699&6.766&9.145&8.216&9.849&9.765&42.430&80.707\\
0.020&18.593&\bfseries 18.365&18.470&31.612&28.553&34.773&34.740&93.195&581.535\\

\bottomrule
\end{tabular}}
\caption{Average query time (in seconds) for answering $10^6$ uniformly sampled queries on $10^7$ points from a normal distribution with varying $10^{-7}\leq\delta\leq 0.02$. Lowest marked \textbf{bold}.}
\label{tab:app:query:scalingw}
\end{table}
\setlength{\tabcolsep}{0.25em}

\begin{table}[H]
{\small
\begin{tabular}{lllrrrrrrrrr}
\toprule
Name&Dist&\competitorBOOSTl&\competitorBOOSTq&\competitorBOOSTR&\competitorEHCDS&\competitorEHCDSHilbert&\competitorHCDS&\competitorHCDSHilbert&\competitorPAM&\competitorPKD\\
 \midrule
athens\_large&norm&0.015&0.015&0.014&\bfseries 0.004&0.005&0.035&0.084&0.549&0.087\\
athens\_large&skew&0.053&0.048&0.050&\bfseries 0.015&0.016&0.048&0.097&0.369&0.294\\
athens\_large&unif&0.031&0.029&0.029&\bfseries 0.008&0.009&0.037&0.090&0.336&0.154\\
athens\_small&norm&0.029&0.027&0.027&\bfseries 0.008&0.011&0.042&0.086&0.105&0.107\\
athens\_small&skew&0.047&0.045&0.045&\bfseries 0.009&0.013&0.032&0.087&0.132&0.202\\
athens\_small&unif&0.039&0.038&0.037&\bfseries 0.008&0.010&0.042&0.087&0.131&0.140\\
berlin\_large&norm&0.035&0.034&0.033&\bfseries 0.014&0.018&0.043&0.099&0.299&0.158\\
berlin\_large&skew&0.134&0.130&0.131&\bfseries 0.058&0.069&0.090&0.140&0.730&0.468\\
berlin\_large&unif&0.081&0.081&0.081&\bfseries 0.036&0.041&0.064&0.115&0.711&0.287\\
chicago&norm&0.016&0.015&0.015&\bfseries 0.006&0.007&0.041&0.092&0.614&0.076\\
chicago&skew&0.085&0.087&0.082&\bfseries 0.031&0.036&0.070&0.118&0.924&0.380\\
chicago&unif&0.042&0.042&0.041&\bfseries 0.015&0.018&0.050&0.102&0.815&0.180\\
drifter&norm&0.013&0.014&0.017&\bfseries 0.009&0.010&0.060&0.104&0.924&0.064\\
drifter&skew&0.135&0.134&0.125&0.086&\bfseries 0.082&0.119&0.180&2.738&0.452\\
drifter&unif&0.100&0.104&0.107&0.064&\bfseries 0.063&0.115&0.153&2.653&0.304\\
tiger2006east&norm&0.167&0.206&0.165&0.179&\bfseries 0.145&0.230&0.263&10.909&0.819\\
tiger2006east&skew&0.234&0.273&0.262&0.225&\bfseries 0.197&0.267&0.331&13.212&1.204\\
tiger2006east&unif&0.250&0.245&0.245&0.255&\bfseries 0.232&0.320&0.360&14.246&1.217\\
tiger2006se&norm&0.111&0.101&0.102&\bfseries 0.021&0.022&0.068&0.119&18.103&0.053\\
tiger2006se&skew&0.919&\bfseries 0.564&0.573&1.236&1.145&1.442&1.554&22.461&29.487\\
tiger2006se&unif&0.296&\bfseries 0.238&0.240&0.366&0.380&0.460&0.519&20.669&8.626\\
unid\_00&norm&0.060&0.059&0.058&\bfseries 0.033&0.038&0.066&0.116&0.963&0.272\\
unid\_00&skew&0.085&0.084&0.084&\bfseries 0.048&0.052&0.079&0.135&1.336&0.440\\
unid\_00&unif&0.056&0.053&0.053&\bfseries 0.029&0.032&0.068&0.117&1.287&0.268\\

\bottomrule
\end{tabular}}
\caption{Average query time (in seconds) for answering $10^6$ queries on realworld data sets for $\delta=0.001$. Lowest marked \textbf{bold}.}
\label{tab:app:realworld:query}
\end{table}

\begin{table}[H]
{\small
\begin{tabular}{lllrrrrrrrrr}
\toprule
Name&\competitorBOOSTl&\competitorBOOSTq&\competitorBOOSTR&\competitorEHCDS&\competitorEHCDSHilbert&\competitorHCDS&\competitorHCDSHilbert&\competitorPAM&\competitorPKD\\
 \midrule
athens\_large&0.008&0.008&0.007&0.010&0.012&\bfseries 0.005&0.009&0.070&0.013\\
athens\_small&2.282e-04&2.329e-04&2.299e-04&0.002&0.002&\bfseries 1.756e-04&3.199e-04&0.002&0.001\\
berlin\_large&0.023&0.023&0.024&0.023&0.032&\bfseries 0.015&0.025&0.082&0.043\\
chicago&0.014&0.014&0.014&0.014&0.020&\bfseries 0.010&0.016&0.122&0.024\\
drifter&0.187&0.190&0.191&0.192&0.275&\bfseries 0.119&0.207&2.190&0.169\\
tiger2006east&2.544&2.565&2.533&1.768&2.550&\bfseries 1.618&2.423&24.951&2.757\\
tiger2006se&5.279&5.253&5.239&3.638&5.145&\bfseries 3.394&5.064&51.696&4.912\\
unid\_00&0.020&0.020&0.020&0.021&0.030&\bfseries 0.013&0.023&0.207&0.029\\

\bottomrule
\end{tabular}}
\caption{Average construction time (in seconds) for answering $10^6$ queries on real-world data sets for $\delta=0.001$. Lowest marked \textbf{bold}.}
\label{tab:app:realworld:const}
\end{table}

\newpage
\subsection{Skewed point generation- uniform sampling}
\begin{figure}
    \centering
    \includegraphics{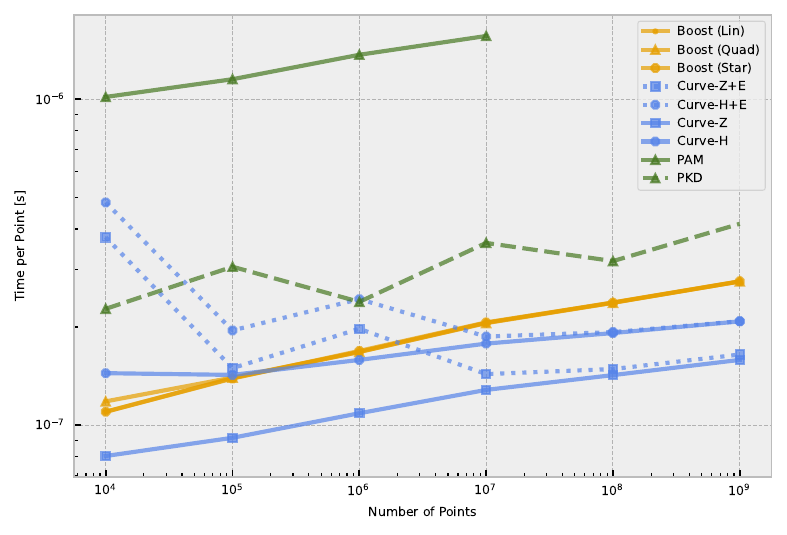}
    \caption{Time needed for constructing the data structure in seconds per point from $10^4$ to $10^9$ (Skewed distribution). }
    \label{fig:exp:norm:unif:scaling:constructionperpoint:skew}
\end{figure}
\begin{figure}
    \centering
    \includegraphics{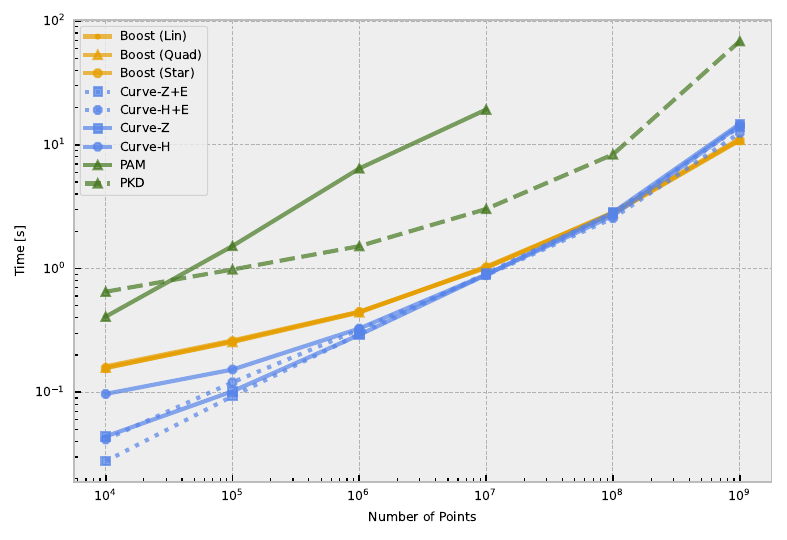}
    \caption{Overall query time to answer  $10^6$ uniform queries with size $\delta=0.001$ for $10^{4}$ to $10^9$ points (Skewed distribution).}
    \label{fig:exp:norm:unif:scaling:query:skew}
\end{figure}
\begin{figure}
    \centering
    \includegraphics{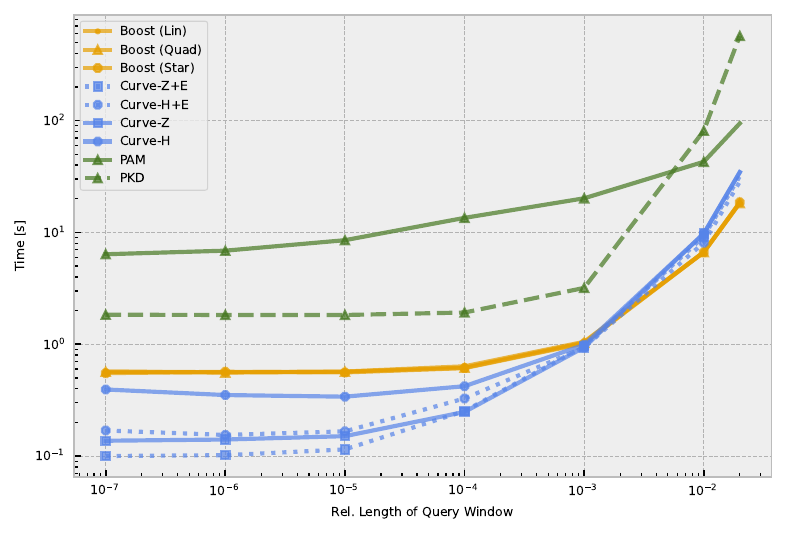}
    \caption{Plot of query time dependent on the relative length of the query window  in seconds (on $10^7$ points from skewed distribution).}
    \label{fig:exp:scalingw:query:skew}
\end{figure}
\begin{table}[H]
{\small
\begin{tabular}{lrrrrrrrrr}
\toprule
&\competitorBOOSTl&\competitorBOOSTq&\competitorBOOSTR&\competitorEHCDS&\competitorEHCDSHilbert&\competitorHCDS&\competitorHCDSHilbert&\competitorPAM&\competitorPKD\\
 \midrule
$10^{4}$&0.001&0.001&0.001&0.004&0.005&\bfseries 8.047e-04&0.001&0.010&0.002\\
$10^{5}$&0.014&0.014&0.014&0.015&0.020&\bfseries 0.009&0.014&0.115&0.031\\
$10^{6}$&0.168&0.168&0.169&0.198&0.244&\bfseries 0.109&0.159&1.369&0.239\\
$10^{7}$&2.061&2.057&2.067&1.435&1.871&\bfseries 1.283&1.780&15.633&3.625\\
$10^{8}$&23.804&23.739&23.767&14.872&19.305&\bfseries 14.248&19.176&\em OOM&31.882\\
$10^{9}$&276.187&276.093&276.312&164.772&208.486&\bfseries 158.643&208.309&\em OOM&414.898\\

\bottomrule
\end{tabular}}
\caption{Average construction time for $10^i$ for $4\leq i\leq 9$.  Skewed distribution of input points. \competitorPAM crashed for $10^8$ and $10^9$. Lowest marked \textbf{bold}.}
\label{tab:app:const:scaling:skewed}
\end{table}

\begin{table}[H]
{\small
\begin{tabular}{lrrrrrrrrr}
\toprule
&\competitorBOOSTl&\competitorBOOSTq&\competitorBOOSTR&\competitorEHCDS&\competitorEHCDSHilbert&\competitorHCDS&\competitorHCDSHilbert&\competitorPAM&\competitorPKD\\
 \midrule
$10^{4}$&0.162&0.158&0.157&\bfseries 0.028&0.042&0.044&0.097&0.408&0.649\\
$10^{5}$&0.262&0.255&0.258&\bfseries 0.093&0.120&0.102&0.153&1.520&0.982\\
$10^{6}$&0.448&0.441&0.446&0.295&0.321&\bfseries 0.291&0.327&6.433&1.516\\
$10^{7}$&1.013&1.036&1.021&0.913&\bfseries 0.882&0.890&0.893&19.246&3.025\\
$10^{8}$&2.849&2.806&2.778&2.810&\bfseries 2.553&2.833&2.675&\em OOM&8.347\\
$10^{9}$&11.186&10.955&\bfseries 10.853&13.980&12.454&14.744&14.288&\em OOM&68.678\\

\bottomrule
\end{tabular}
}
\caption{Average query time for answering $10^6$ queries on $10^i$ points for $4\leq i\leq 9$.  Skewed distribution of input and uniform random querying. \competitorPAM crashed for $10^8$ and $10^9$. Lowest marked \textbf{bold}.}
\label{tab:app:query:scaling:skew}
\end{table}
\begin{table}[H]
{\small
\begin{tabular}{lrrrrrrrrr}
\toprule
$\delta$&\competitorBOOSTl&\competitorBOOSTq&\competitorBOOSTR&\competitorEHCDS&\competitorEHCDSHilbert&\competitorHCDS&\competitorHCDSHilbert&\competitorPAM&\competitorPKD\\
 \midrule
1.000e-07&0.569&0.569&0.552&\bfseries 0.100&0.170&0.137&0.394&6.381&1.837\\
1.000e-06&0.562&0.559&0.568&\bfseries 0.102&0.155&0.141&0.352&6.869&1.828\\
1.000e-05&0.570&0.569&0.560&\bfseries 0.115&0.166&0.151&0.340&8.524&1.828\\
1.000e-04&0.635&0.617&0.612&0.253&0.328&\bfseries 0.248&0.422&13.497&1.923\\
0.001&1.051&1.033&1.018&0.960&\bfseries 0.936&0.938&0.989&20.196&3.191\\
0.010&6.657&6.672&\bfseries 6.649&9.188&8.104&9.770&9.756&42.819&81.108\\
0.020&18.407&\bfseries 18.304&18.633&31.807&28.406&34.642&34.419&94.418&570.746\\

\bottomrule
\end{tabular}}
\caption{Average query time (in seconds) for answering $10^6$ uniformly sampled queries on $10^7$ points from a skewed distribution with varying $10^{-7}\leq\delta\leq 0.02$. Lowest marked \textbf{bold}.}
\label{tab:app:query:scalingw:skew}
\end{table}
\setlength{\tabcolsep}{0.25em}
\newpage

\subsection{Uniform point generation- uniform sampling}
\begin{figure}
    \centering
    \includegraphics{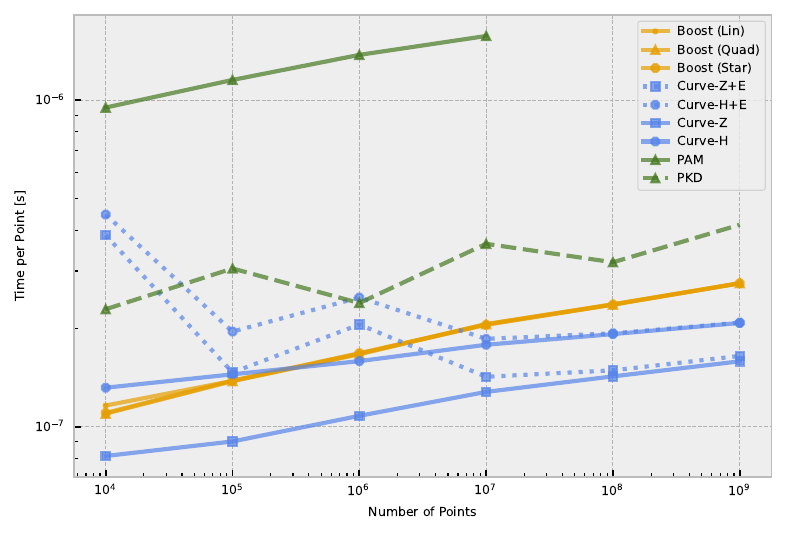}
    \caption{Time needed for constructing the data structure in seconds per point from $10^4$ to $10^9$ (uniform distribution). }
    \label{fig:exp:norm:unif:scaling:constructionperpoint:unif}
\end{figure}
\begin{figure}
    \centering
    \includegraphics{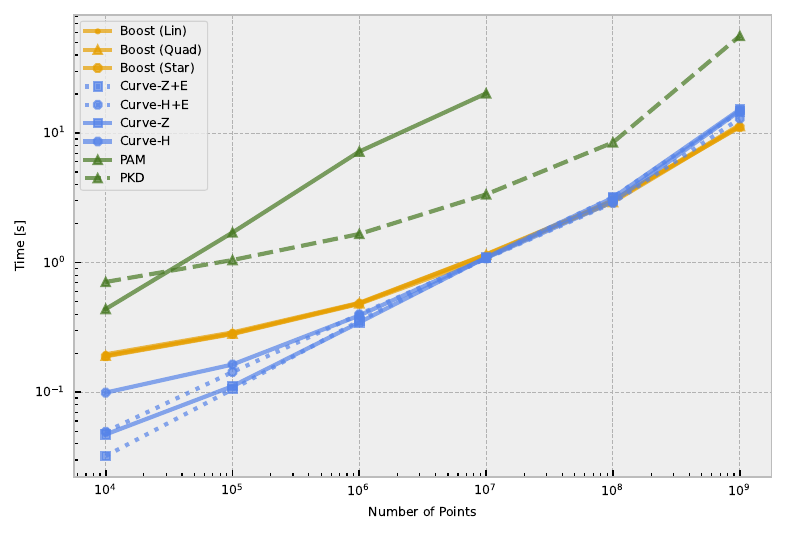}
    \caption{Overall query time to answer  $10^6$ uniform queries with size $\delta=0.001$ for $10^{4}$ to $10^9$ points (Uniform distribution).}
    \label{fig:exp:norm:unif:scaling:query:unif}
\end{figure}
\begin{figure}
    \centering
    \includegraphics{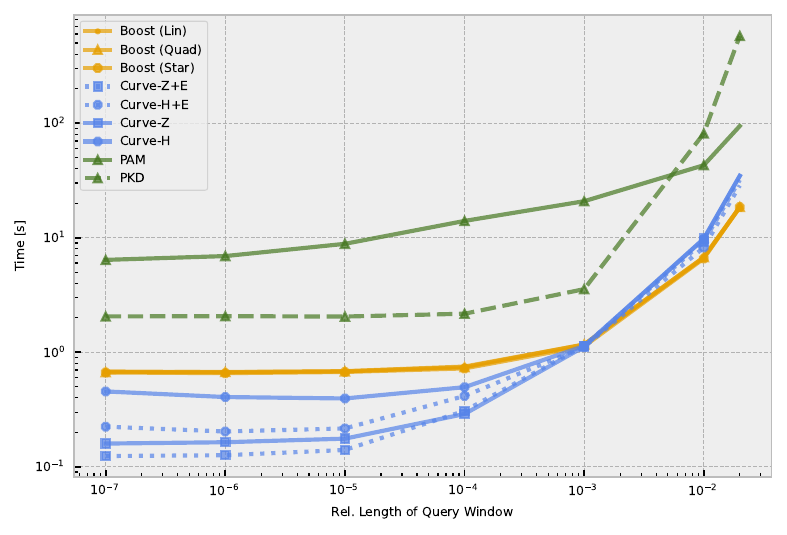}
    \caption{Plot of query time dependent on the relative length of the query window  in seconds (on $10^7$ points from unif distribution).}
    \label{fig:exp:scalingw:query:unif}
\end{figure}
\begin{table}[H]
{\small
\begin{tabular}{lrrrrrrrrr}
\toprule
&\competitorBOOSTl&\competitorBOOSTq&\competitorBOOSTR&\competitorEHCDS&\competitorEHCDSHilbert&\competitorHCDS&\competitorHCDSHilbert&\competitorPAM&\competitorPKD\\
 \midrule
$10^{4}$&0.001&0.001&0.001&0.004&0.004&\bfseries 8.142e-04&0.001&0.009&0.002\\
$10^{5}$&0.014&0.014&0.014&0.015&0.020&\bfseries 0.009&0.014&0.115&0.031\\
$10^{6}$&0.168&0.167&0.168&0.206&0.249&\bfseries 0.108&0.159&1.376&0.239\\
$10^{7}$&2.063&2.061&2.052&1.423&1.859&\bfseries 1.279&1.784&15.725&3.634\\
$10^{8}$&23.732&23.658&23.676&14.894&19.327&\bfseries 14.270&19.225&\em OOM&31.904\\
$10^{9}$&274.948&275.382&275.193&164.830&208.450&\bfseries 158.713&207.996&\em OOM&415.579\\

\bottomrule
\end{tabular}}
\caption{Average construction time for $10^i$ for $4\leq i\leq 9$.  Uniform distribution of input points. \competitorPAM crashed for $10^8$ and $10^9$. Lowest marked \textbf{bold}.}
\label{tab:app:const:scaling:unif}
\end{table}

\begin{table}[H]
{\small
\begin{tabular}{lrrrrrrrrr}
\toprule
&\competitorBOOSTl&\competitorBOOSTq&\competitorBOOSTR&\competitorEHCDS&\competitorEHCDSHilbert&\competitorHCDS&\competitorHCDSHilbert&\competitorPAM&\competitorPKD\\
 \midrule
$10^{4}$&0.196&0.191&0.189&\bfseries 0.032&0.049&0.047&0.099&0.438&0.709\\
$10^{5}$&0.289&0.283&0.281&\bfseries 0.105&0.143&0.111&0.163&1.707&1.047\\
$10^{6}$&0.489&0.485&0.480&0.355&0.399&\bfseries 0.345&0.391&7.198&1.663\\
$10^{7}$&1.159&1.150&1.097&1.084&\bfseries 1.075&1.096&1.086&20.248&3.361\\
$10^{8}$&2.991&2.933&3.042&3.202&\bfseries 2.888&3.165&2.946&\em OOM&8.458\\
$10^{9}$&11.560&11.330&\bfseries 11.140&14.643&12.938&15.188&14.850&\em OOM&56.158\\

\bottomrule
\end{tabular}
}
\caption{Average query time for answering $10^6$ queries on $10^i$ points for $4\leq i\leq 9$.  Uniform distribution of input and uniform random querying. \competitorPAM crashed for $10^8$ and $10^9$. Lowest marked \textbf{bold}.}
\label{tab:app:query:scaling:unif}
\end{table}
\begin{table}[H]
{\small
\begin{tabular}{lrrrrrrrrr}
\toprule
$\delta$&\competitorBOOSTl&\competitorBOOSTq&\competitorBOOSTR&\competitorEHCDS&\competitorEHCDSHilbert&\competitorHCDS&\competitorHCDSHilbert&\competitorPAM&\competitorPKD\\
 \midrule
1.000e-07&0.680&0.670&0.667&\bfseries 0.124&0.224&0.160&0.455&6.401&2.055\\
1.000e-06&0.669&0.672&0.657&\bfseries 0.126&0.203&0.164&0.406&6.911&2.064\\
1.000e-05&0.683&0.681&0.673&\bfseries 0.140&0.216&0.176&0.395&8.839&2.050\\
1.000e-04&0.750&0.740&0.723&0.308&0.415&\bfseries 0.289&0.495&14.005&2.169\\
0.001&1.167&1.168&1.127&1.140&\bfseries 1.112&1.118&1.135&20.815&3.554\\
0.010&6.716&6.751&\bfseries 6.584&9.246&8.208&9.801&9.726&42.891&81.201\\
0.020&\bfseries 18.413&18.518&18.535&31.900&28.539&34.692&34.604&94.350&576.249\\

\bottomrule
\end{tabular}}
\caption{Average query time (in seconds) for answering $10^6$ uniformly sampled queries on $10^7$ points from a uniform distribution with varying $10^{-7}\leq\delta\leq 0.02$. Lowest marked \textbf{bold}.}
\label{tab:app:query:scalingw:unif}
\end{table}
\end{document}